\RequirePackage{fix-cm}
\documentclass[smallextended]{svjour3}       
\smartqed  
\usepackage{graphicx}
\usepackage{graphicx}
\usepackage{graphicx}
\usepackage{amsfonts}
\usepackage[none]{hyphenat}
\usepackage{pdfpages}
\usepackage{graphicx}
\usepackage{ragged2e}
\usepackage{float}
\usepackage{comment}
\begin{document}

\title{Lie symmetry analysis and explicit solutions for the time fractional generalized Burgers-Fisher Equation
}


\author{Ramya Selvaraj$^1$        \and
       V.Swaminathan$^1$  \and A.Durga Devi$^2$  \and K.Krishnakumar$^1$  
}


\institute{$^1$ 
             Department of Mathematics,
             Srinivasa Ramanujan Centre, 
              SASTRA Deemed to be University,
              Kumbakonam 612 001, India.
             \email{rsramyaselvaraj@gmail.com, mvsnew@gmail.com}             \\
                $^2$ 
             Department of Physics,
             Srinivasa Ramanujan Centre, 
              SASTRA Deemed to be University,
              Kumbakonam 612 001, India.
}
\date{Received: date / Accepted: date}

\maketitle

\begin{abstract}
In this article,  we study the Lie point symmetries for the time fractional generalized Burgers-Fisher (GBF) equation. While getting an appropriate combination of symmetries, the time fractional partial differential equation  has been transformed to nonlinear fractional ordinary differential equations (ODE) using Erdelyi-Kober differential operator. Furthermore, using power series method, we get the exact solution of the nonlinear fractional GBF equation with the arbitrary nonlinearity.

\end{abstract}
\section{Introduction}
\hspace{\parindent} The fractional differential equation (FDE) plays a vital role in many branches of science and engineering \cite{kilbas,juma,eslami,podlubny,huan}. FDE has many applications in the field of magnetism, fluid mechanics, cardiac tissue electrode interface, ultrasonic wave propagation in human cancellous bone, RLC electric circuit, theory of visco elasticity, lateral and longitudinal control of autonomous vehicles, wave propagation in viscoelastic horn, heat transfer, sound waves propagation in rigid porous materials and many more.
    
    In recent years, many authors find solutions of FDEs using various methods such as variational iteration method \cite{wu},  homotopy perturbation method \cite{Gepreelhomotopy}, Adomian decomposition method \cite{elsayed}, the first integral method \cite{bekirunsal,Lu},  the sub-equation method \cite{Guo,tong,zhangzhang} and so on.
    
    Although a large number of methods have been developed to solve FDEs \cite{Alzaidy,bulutbaskonus,bulutbernoulli,bulutpandir,tarasov,veeresha}, in recent years, to reach the exact solutions of nonlinear partial differential equations \cite{kirya}, Lie symmetry method is used which is considered as an efficient method. In the beginning of 19th century, the Norwegian mathematician Sophus Lie \cite{lie} introduced the Lie symmetry analysis.
    
    To find one or several parameter continuous transformations leaving the equation invariant, the Lie symmetry analysis is used. Later, it has been developed by Ovsiannikov \cite{ovsiannikov} and many researchers used it to solve various equations  \cite{bal18,bal17,juma,bluman,hydon,jeff,kkr,cradd,mustafa,ibragi,oldham,kkr2,olver,kkr1}. 
    
    The time fractional KdV equations using Lie group analysis was explained by Wang and Xu \cite{wang}. He also explained the invariant analysis and explicit solutions of the time fractional nonlinear perturbed Burgers equation \cite{wang2} and Y. W. Zhang \cite{zhangyw} performed Lie symmetry analysis to the time fractional generalized fifth order KdV equation. 
    
    V. Kumar et al \cite{kumar} studied Lie symmetry based analytical and numerical approach for modified Burgers-KdV equation. Furthermore, T. Bakkyaraj and R. Sahadevan \cite{baakkiyaraj,saha} determined the invariant analysis of nonlinear fractional ordinary differential equations with Riemann-Liouville fractional derivative. They also performed group formalism of Lie transformations to time fractional differential equations \cite{baksaha}.

This  work is organized as follows: In section 2, we give some preliminaries on fractional derivatives. In section 3 and 4, we study the  Lie point symmetries and symmetry reduction of the GBF equation. In section 5 and 6, we  derive the  power series solution to find the  explicit  solutions of the resultant equation and analyze convergence. Physical meaning of the exact solution of the power series is given in section 7. Section 8 ends with conclusion.

\section{Preliminaries}
    Many researchers  used several definitions of fractional derivative such as  the Caputo \cite{kilbas,podlubny},  Riemann-Liouville \cite{kilbas,oldham,podlubny,samko}, the Weyl \cite{kilbas,samko}, the Grunwald-Letnikov \cite{kilbas,oldham,podlubny,samko} and the Riesz \cite{samko}. Among them Caputo and Riemann-Liouville fractional derivatives have been widely used. 

The Riemann-Liouville (RL) fractional  derivative is used to study the Lie symmetry analysis of fractional differential equations \cite{kilbas,hydon,ibragi,podlubny,baakkiyaraj,saha}.  In this paper, we use some basic definitions.
\subsection{Definition}
\label{sec:1}
The RL fractional derivative \cite{gazizov,huang,baakkiyaraj} is given by
\begin{equation}
D^{\alpha}f(t)= \left\{\begin{array}{ll}
\displaystyle{\frac{d^nf}{dt^n}} & \alpha = n,\\
 \displaystyle{\frac{d^n}{dt^n}I^{n-\alpha}f(t)}, & 0 \leq n-1 < \alpha < n, \end{array} \right.
\end{equation}
given by $n\in \mathbb{N}$, $I^\mu$ f(t) is the RL fractional integral of order $\mu$, where
\begin{center}
\begin{eqnarray}
I^\mu f(t)&=&\frac{1}{\Gamma(\mu)}\int_0^t(t-s)^{\mu-1}f(s)ds, \mu>0 \nonumber \\
\end{eqnarray}
\end{center}
and $\Gamma(z)$ is the gamma function.

\subsection{Definition}\label{sec:2}
The RL  partial fractional derivative is given by
\begin{equation} \label{RLPDE}
\partial^\alpha_t =\left\{\begin{array}{ll}
\frac{\partial^nf}{\partial t^n},& \alpha = n, \\
\frac{1}{\Gamma (n-\alpha)}\frac{\partial^n}{\partial t^n}\int_0^t(t-s)^{n-\alpha-1} u(s,x)ds, & 0 \leq n-1 < \alpha < n,n\in N. \end{array} \right.
\end{equation}
If it exists, where $\partial^\alpha_t$ is the partial derivative of integer $n$ \cite{gazizov,huang,baakkiyaraj}.

\section{Method of Lie symmetry for the time fractional differential equations}
Let us consider the time fractional partial differential equation (TF-PDE) having the form 
\begin{equation} \label{TFPDE}
\partial ^\alpha_t u=F(t,x,u,u_x,u_{xx},\ldots ),\; (0<\alpha<1).
\end{equation}
A one-parameter Lie group of transformations are given by
\begin{eqnarray}\label{oneparameter}
\tilde{t}&=&t+\epsilon \zeta(t,x,u)+O(\epsilon^2),\nonumber\\  
\tilde{x}&=&x+\epsilon \xi(t,x,u)+O(\epsilon^2),\nonumber\\ 
\tilde{u}&=&u+\epsilon \eta(t,x,u)+O(\epsilon^2),\nonumber\\ 
\frac{\partial^\alpha \tilde{u}}{\partial \tilde{t}^\alpha}&=&\frac{\partial^\alpha u}{\partial t^\alpha} +\epsilon \eta^0_\alpha (t,x,u) + O(\epsilon^2),\\
\frac{\partial \tilde{u}}{\partial \tilde{x}}&=&\frac{\partial u}{\partial x}+\epsilon \eta^x (t,x,u)+O(\epsilon^2), \nonumber\\ 
\frac{\partial^2 \tilde{u}}{\partial \tilde{x^2}}&=&\frac{\partial^2 u}{\partial x^2}+\epsilon \eta^{xx}(t,x,u)+O(\epsilon^2), \nonumber  
\end{eqnarray}
where
\begin{eqnarray}\label{etavalues}
\eta^x&=&D_x (\eta)-u_x D_x (\xi)-u_t D_x (\zeta), \nonumber\\
\eta^{xx}&=&D_x (\eta^x)-u_{xt} D_x(\xi)-u_{xx}D_t(\zeta),\\ \nonumber
\end{eqnarray}
and the total differential operator $D_x$ is defined by
\begin{equation}
D_x=\frac{\partial}{\partial x}+u_x \frac{\partial}{\partial u}+u_{xx} \frac{\partial}{\partial u_x}+...\\ 
\end{equation}
The associated Lie algebra of symmetries is spanned by vector fields 
\begin{equation}\label{totop}
X= \xi \frac{\partial}{\partial x}+\zeta \frac{\partial}{\partial t}+\eta \frac{\partial}{\partial u}.
\end{equation}
The vector field Eq.(\ref{totop}) is a Lie point symmetry of Eq.(\ref{RLPDE}) provided
\begin{equation}
pr^{\alpha,2}X(\nabla)|_{\nabla=0}=0.
\end{equation}
Also, the invariance condition gives 
\begin{equation}\label{Invcon}
\zeta (t,x,u)|_{t=0}=0,
\end{equation}
and the $\alpha^{th}$ extended infinitesimal related to RL fractional time derivative with Eq.(\ref{Invcon}) is given by
\begin{eqnarray} \label{extendinfinite}
\eta^0_\alpha=\frac{\partial^\alpha \eta}{\partial t^\alpha}+(\eta_u-\alpha D_t(\zeta))\frac{\partial^\alpha u}{\partial t^\alpha}+\mu -\sum^\infty_{n=1}{\alpha \choose n}D^n_t(\xi)D^{\alpha- n}_t(u_x)\nonumber\\
+\sum^\infty_{n=1}\Bigg[{\alpha \choose n}\frac{\partial^\alpha \eta_u}{\partial t^\alpha}-{\alpha \choose n+1}D^{n+1}_t(\zeta)\Bigg]D^{\alpha-n}_t(u),
\end{eqnarray}
where
\begin{eqnarray} \label{m}
\mu&=&\sum^\infty_{n=2}\sum^{n}_{m=2}\sum^{m}_{k=2}\sum^{k-1}_{r=0}{\alpha \choose n}{n \choose m}{k \choose r}\frac{1}{k!}\frac{t^{n-\alpha}}{\Gamma(n+1-\alpha)}[-u]^r \nonumber \\
&&\frac{\partial^m}{\partial t^m}[u^{k-r}]\frac{\partial^{n-m+k}\eta}{\partial t^{n-m}\partial u^k}.
\end{eqnarray}
 Due to the presence of $\frac{\partial \eta^k}{\partial u^k}$, if the infinitesimal $\eta$ is linear in u, the expression for  $\mu=0$ for $ k\geq 2$ in Eq.(\ref{m}).
\subsection{Definition}
The function $u=\theta (x,t)$ is an invariant solution of Eq.(\ref{TFPDE}) associated with Eq.(\ref{totop}) such that 
\begin{enumerate}
\item $u=\theta(x,t)$ satisfies Eq.(\ref{TFPDE}).
\item $u=\theta(x,t)$ is an invariant surface of Eq.(\ref{oneparameter}), this means that
\begin{center}
$\zeta(t,x,\theta)\theta_t+\xi(t,x,\theta)\theta_x=\eta(t,x,\theta)$.
\end{center}
\end{enumerate}

\section{Lie symmetries for time fractional generalized Burgers-Fisher equation }
\hspace{\parindent} In this work, Lie symmetry method has been presented for time fractional generalized Burgers-Fisher (GBF) equation.
 
    The mathematical modelling of turbulence was explained by a Dutch physicist, Johannus Martinus Burgers, in 1948.  A nonlinear equation which is the combination of reaction, convection and diffusion mechanism is called Burgers-Fisher equation.

The  GBF equation is used in the field of fluid dynamics. It has also been found in some applications such as gas dynamics, heat conduction, elasticity and so on.

The time fractional GBF equation is given by, 
\begin{equation} \label{GBF}
u^\alpha _t+\beta u^\delta u_x-u_{xx}=\gamma u(1-u^\delta) 
\end{equation}
where $0<\alpha \leq 1$,  $\alpha$ is the order of fractional time derivative and $\beta,\gamma,\delta$ are arbitrary constants.

Let us consider Eq.(\ref{GBF}) is invariant with respect to Eq.(\ref{oneparameter}), we have that
\begin{equation}\label{tildeGBF}
\tilde{u}^\alpha_{\tilde{t}} + \beta \tilde{u}^\delta \tilde{u}_{\tilde{x}}-\tilde{u}_{\tilde{x} \tilde{x}}=\gamma \tilde{u}(1-\tilde{u}^\delta),
\end{equation}
such that $\tilde{u}=u(\tilde{x},\tilde{t})$ satisfies Eq.(\ref{GBF}). Using  Eq.(\ref{oneparameter}) in Eq.(\ref{tildeGBF}), we get the invariant equation
\begin{equation} \label{invariantequation}
\eta^0_\alpha +\beta\delta u^{\delta-1}\eta u_x+\beta u^\delta \eta^x-\eta^{xx}-\gamma\eta +\gamma(\delta+1)u^\delta\eta=0.
\end{equation}
Applying the values of $\eta^0_\alpha,\eta^x$ and $\eta^{xx}$ given in Eq.(\ref{etavalues}) and Eq.(\ref{extendinfinite}) into Eq.(\ref{invariantequation}) and then isolating coefficients in partial derivatives with respect to $x$ and power of $u$, we get
\begin{eqnarray}
\partial^\alpha_t \eta - u\partial^\alpha_t\eta_u+\eta\gamma+\eta_{xx}-\gamma(\delta+1)u^\delta-\beta u^\delta\eta_x = 0, \\ \nonumber
{\alpha \choose n}\partial^n_t(\eta)-{\alpha \choose n+1}D^{n+1}_t(\zeta)=0, n=1,2,...\\ \nonumber
\xi_u=\zeta_u=\zeta_t=\zeta_x= 0,\\ \nonumber
\eta_{uu}=\zeta_{uu}=\xi_{uu}=0. \\ \nonumber
\end{eqnarray}
Solving the over determining equations, we get:
\begin{center}
$\zeta=k_1+x\alpha k_2,\;  \xi=2t\delta k_2, \; \eta = -\alpha u k_2$,
\end{center}
where $k_1$ and $k_2$ are arbitrary constants. Thus infinitesimal symmetry group for Eq.(\ref{GBF}) is spanned by the two vector fields
\begin{equation}
X_1=\frac{\partial}{\partial x},\; X_2=2t\frac{\partial}{\partial t}+x\alpha \delta \frac{\partial}{\partial x}-u\alpha \frac{\partial}{\partial u}.
\end{equation}
In particular, the symmetry $X_2$ possess the similarity transformation and similarity variable as follows:
\begin{equation}
k_1=x t^{\frac{-\alpha}{2}}, \; k_2=u t^{\frac{\alpha}{2\delta}},
\end{equation}
and this yields
\begin{equation}\label{symm}
u=t^{\frac{-\alpha}{2\delta}}f(\xi), \; \xi =x t^\frac{-\alpha}{2}. 
\end{equation}
In Eq.(\ref{symm}), $f$ is an arbitrary function of $\xi$. Using Eq.(\ref{symm}), Eq.(\ref{GBF}) is transformed to a special nonlinear ODE of fractional order as mentioned in the following theorem 1.

\begin{theorem}
The similarity transformation $u(x, t) = t^{\frac{-\alpha}{2\delta}}f(\xi)$  along with the similarity variable $\xi =x t^\frac{-\alpha}{2}$ reduces Eq.(\ref{GBF}) to the nonlinear fractional ODE of the form,  
\begin{equation} \label{reduced eqn}
\Bigg(P^{1-\frac{\alpha (1
+2\delta)}{2\delta},\; \alpha}_{\frac{2}{\alpha}} f \Bigg)(\xi)+\beta f^\delta f_\xi-u_{xx}-\gamma u(1-u^\delta)=0,
\end{equation}
where the Erdelyi-Kober fractional differential operator (EK-FDO)  \cite{kirya}
\begin{equation} \label{Pvalue}
\Big(P^{\zeta,\alpha}_\beta f \Big) (\xi) = \prod ^{n-1}_{j=0} \Big(\zeta+j-\frac{1}{\beta}\frac{d}{d\xi}\Big)(K^{\zeta+\alpha,n-\alpha}_\beta f)(\xi),
\end{equation}
\begin{equation}
 n=\left\{\begin{array}{ll}
[\alpha]+1, & \alpha \neq \mathbb{N},\\
\alpha , & \alpha \in \mathbb{N},
\end{array}\right.
\end{equation}
where
\begin{equation}\label{EKoperator}
(K^{\zeta,\alpha}_\beta f)(\xi)= \left\{ \begin{array}{ll}
\frac{1}{\Gamma (\alpha)}\int_ 1^\infty (u-1)^{\alpha -1}u^{-(\zeta+\alpha)}f(\xi u^{\frac{1}{\beta}})du,& \alpha >0,\\
f(\xi), & \alpha =0,
\end{array}\right.
\end{equation}
is the Erdelyi-Kober fractional integral operator (EK-FIO) \cite{baakkiyaraj}.
\end{theorem}

\begin{proof}
 Let $n-1<\alpha<n, \; n=1,2,3,...$ Then the RL fractional derivative for the similarity transformation of Eq.(\ref{symm}) becomes
\begin{equation} \label{RNODE}
\frac{\partial^\alpha u}{\partial t^\alpha}=\frac{\partial^n}{\partial t^n}\Bigg[\frac{1}{\Gamma(n-\alpha)}\int_0^t (t-s)^{n-\alpha-1}s^{-\frac{\alpha}{2\delta}} f(xs^{-(\frac{\alpha}{2})})ds\Bigg].
\end{equation}
Let $v=\frac{t}{s},  ds=-\frac{t}{v^2}dv.$ Thus, Eq.(\ref{RNODE}) becomes
\begin{equation} \label{aplly EK}
\frac{\partial^\alpha u}{\partial t^\alpha}=\frac{\partial^n}{\partial t^n}\Bigg[t^{n-\frac{\alpha(1+2\delta)}{2\delta}}\frac{1}{\Gamma(n-\alpha)}\int_1^{\infty}(v-1)^{n-\alpha-1}v^{-(n+1-\frac{\alpha(1+2\delta)}{2\delta})}f(\xi v^{\frac{\alpha}{2}})dv \Bigg].
\end{equation}
Applying EK-FDO Eq.(\ref{EKoperator}) in Eq.(\ref{aplly EK}), we have
\begin{equation}\label{newEK}
\frac{\partial^\alpha u}{\partial t^\alpha}=\frac{\partial^n}{\partial t^n}\Big[t^{n-\frac{\alpha(1+2\delta)}{2\delta}}\Big(K_\frac{2}{\alpha}^{1-\frac{\alpha}{2\delta},n-\alpha}f\Big)(\xi)\Big].
\end{equation}
In order to  simplify the right hand side of Eq.(\ref{newEK}), let us consider the relation 
 $\xi=xt^{-\frac{\alpha}{2}}, \varphi\in (0,\infty)$, we acquire
\begin{equation}
t\frac{\partial}{\partial t}\varphi(\xi)=tx(-\frac{\alpha}{2})t^{-\frac{\alpha}{2}-1}\varphi{\prime}(\xi)=-{\frac{\alpha}{2}}\xi\frac{\partial}{\partial\xi}\varphi(\xi).
\end{equation}
Hence,
\begin{eqnarray}
&&\frac{\partial^n}{\partial t^n}\Big[t^{n-\frac{\alpha(1+2\delta)}{2\delta}}\Big(K_\frac{2}{\alpha}^{1-\frac{\alpha}{2\delta},n-\alpha}f\Big)(\xi)\Big] \nonumber \\&=&\frac{\partial ^{n-1}}{\partial t^{n-1}}\Bigg[\frac{\partial}{\partial t}\Big(t^{n-\frac{\alpha(1+2\delta)}{2\delta}}\Big(K_{\frac{2} {\delta}}^{1-{\frac{\alpha}{2\delta}},{n-\alpha}}f \Big)(\xi)\Big)\Bigg] \nonumber\\
&=& \frac{\partial ^{n-1}}{\partial t^{n-1}}\Bigg[t^{n-\frac{\alpha(1+2\delta)}{2\delta}-1}\Big(n-{\frac{\alpha}{2\delta}}+\alpha-{\frac{\alpha}{2}}\xi {\frac{\partial}{\partial \xi}} \Big)\Big(K_{\frac{2} {\delta}}^{1-{\frac{\alpha}{2\delta}},{n-\alpha}}f \Big)(\xi)\Big) \Bigg].
\end{eqnarray}
Processing repeatedly for $n-1$ times, we get
\begin{eqnarray} \label{repeating}
&&\frac{\partial^n}{\partial t^n}\Big[t^{n-\frac{\alpha(1+2\delta)}{2\delta}}\Big(K_\frac{2}{\alpha}^{1-\frac{\alpha}{2\delta},n-\alpha}f\Big)(\xi)\Big]\nonumber\\
&=&t^{-\frac{\alpha(1+2\delta)}{2\delta}}\Bigg[\Big(1-\frac{\alpha(1+2\delta)}{2\delta}+j-{\frac{\alpha}{2}}\xi {\frac{\partial}{\partial \xi}}\Big(K_{\frac{2} {\delta}}^{1-{\frac{\alpha}{2\delta}},{n-\alpha}}f \Big)(\xi) \Bigg].
\end{eqnarray}
Applying EK-FDO Eq.(\ref{Pvalue}) in Eq (\ref{repeating}), we have
\begin{equation} \label{nth EK}
\frac{\partial^n u}{\partial t^n}\Big[t^{n-\frac{\alpha(1+2\delta)}{2\delta}}\Big(K_\frac{2}{\alpha}^{1-\frac{\alpha}{2\delta},n-\alpha}f\Big)(\xi)\Big]=t^{-{\frac{\alpha(1+2\delta)}{2\delta}}}\Bigg(P_{\frac{2}{\alpha}}^ {1-\frac{\alpha (1+2\delta)}{2\delta},\alpha}f \Bigg)(\xi).
\end{equation}
Using Eq.(\ref{nth EK}) into Eq.(\ref{newEK}), we obtain
\begin{equation}
\frac{\partial^n u}{\partial t^n}=t^{-{\frac{\alpha(1+2\delta)}{2\delta}}}\Bigg(P_{\frac{2}{\alpha}}^ {1-\frac{\alpha (1+2\delta)}{2\delta},\alpha}f \Bigg)(\xi).
\end{equation}
Thus Eq.(\ref{GBF}) becomes
\begin{equation}\label{frac order ODE}
\Bigg(P_{\frac{2}{\alpha}}^ {1-\frac{\alpha (1+2\delta)}{2\delta},\alpha}f \Bigg)(\xi)+\beta f^\delta f_\xi-f_{\xi \xi}-\gamma f(1-f^\delta)=0.
\end{equation}
\end{proof}

\section{Explicit power series solutions}
In this section, we obtain the power series method for the resultant equation Eq.(\ref{frac order ODE}) which has the arbitrary nonlinearity $\delta$. Furthermore, we analyze convergence of the exact power series solution which is differentiable. \\
Let us assume that
\begin{eqnarray}\label{liepower1}
 f(\xi)&=&\sum_{n=0}^{\infty}b_n\xi ^n,\\
    f^{\prime}(\xi)&=&\sum_{n=0}^{\infty}(n+1)b_{n+1} \xi^{n},\label{liepower2}\\ f^{\prime\prime}(\xi)&=&\sum_{n=0}^{\infty}(n+2)(n+1)b_{n+2} \xi^{n}.\label{liepower3}
\end{eqnarray}
Substituting Eqs.(\ref{liepower1}),(\ref{liepower2}) and (\ref{liepower3}) into Eq.(\ref{frac order ODE}), we get
\begin{eqnarray}
&&\sum_{n=0}^{\infty}\frac{\Gamma(2-\frac{\alpha(1+2\delta)}{2\delta}+\frac{n\alpha}{2})}{\Gamma(2-\frac{\alpha(1+2\delta)}{2\delta}+\frac{n\alpha}{2}+\alpha)}b_n\xi^n+\beta \Big(\sum_{n=0}^{\infty}b_n\xi^n \Big)^{\delta} \sum_{n=0}^{\infty}(n+1)b_{n+1}\xi^n \nonumber \\
&& -\sum_{n=0}^{\infty} (n+2)(n+1)b_{n+2}\xi^n-\gamma\sum_{n=0}^{\infty}b_n\xi^n +\gamma \Big(\sum_{n=0}^{\infty}b_n \xi ^n \Big)^{\delta+1}=0.
\end{eqnarray}
Expanding  $\delta$ times, we obtain

\begin{eqnarray} \label{deltatimesexpansion}
&& \sum_{n=0}^{\infty}\frac{\Gamma(2-\frac{\alpha(1+2\delta)}{2\delta}+\frac{n\alpha}{2})}{\Gamma(2-\frac{\alpha(1+2\delta)}{2\delta}+\frac{n\alpha}{2}+\alpha)}b_n\xi^n+\beta \Big(\sum_{n=0}^{\infty}\sum_{k_1=0}^{n}\sum_{k_2=0}^{k_1}...\sum_{k_{\delta-1}=0}^{k_{\delta-2}}\nonumber\\
&&\sum_{k_\delta=0}^{k_{\delta-1}}{b_{k_\delta}} {b_{{k_{\delta-1}}-{k_{\delta}}}}  
 {b_{{k_{\delta-2}}-{k_{\delta-1}}}}...{b_{{k_{2}}-{k_{3}}}}{b_{{k_{1}}-{k_{2}}}} (n-k_1+1){b_{(n-k_1+1)}} \Big) {\xi^n} \nonumber\\ 
&& -\sum_{n=0}^{\infty} (n+2)(n+1)b_{n+2}\xi^n-\gamma\sum_{n=0}^{\infty}b_n\xi^n +\gamma \Big(\sum_{n=0}^{\infty}\sum_{k_1=0}^{n}\sum_{k_2=0}^{k_1}...\sum_{k_{\delta-1}=0}^{k_{\delta-2}}\nonumber  \\  
&&\sum_{k_\delta=0}^{k_{\delta-1}}{b_{k_\delta}} 
{b_{{k_{\delta-1}}-{k_{\delta}}}} {b_{{k_{\delta-2}}-{k_{\delta-1}}}}...{b_{{k_{2}}-{k_{3}}}}{b_{{k_{1}}-{k_{2}}}} {b_{n-k_1}} \Big)\xi^n= 0.
\end{eqnarray}
Comparing coefficients in Eq.(\ref{deltatimesexpansion}), when $n=0$, we have
\begin{equation}
    b_2=\frac{1}{2}\Bigg[ \frac{\Gamma(2-\frac{\alpha(1+2\delta)}{2\delta})}{\Gamma(2-\frac{\alpha(1+2\delta)}{2\delta}+\alpha)}b_0 + 
    \beta b_1 b_0^{\delta}-\gamma b_0 + \gamma b_0^{\delta+1}\Bigg].
\end{equation}
When $n\geq1$, we have the recurrence relations among the coefficients become 
\begin{eqnarray} \label{bnplustwo} b_{n+2}&=&\frac{1}{(n+1)(n+2)}\Bigg[\sum_{n=0}^{\infty}\frac{\Gamma(2-\frac{\alpha(1+2\delta)}{2\delta}+\frac{n\alpha}{2})}{\Gamma(2-\frac{\alpha(1+2\delta)}{2\delta}+\frac{n\alpha}{2}+\alpha)}b_n \nonumber \\
&&+\beta \Big(\sum_{k_1=0}^{n}\sum_{k_2=0}^{k_1}... 
\sum_{k_{\delta-1}=0}^{k_{\delta-2}}\sum_{k_\delta=0}^{k_{\delta-1}}{b_{k_\delta}} {b_{{k_{\delta-1}}-{k_{\delta}}}}  
{b_{{k_{\delta-2}}-{k_{\delta-1}}}}...{b_{{k_{2}}-{k_{3}}}}\nonumber\\
&&{b_{{k_{1}}-{k_{2}}}} (n-k_1+1){b_{(n-k_1+1)}} \Big) -\gamma b_n +\gamma \Big(\sum_{k_1=0}^{n}\sum_{k_2=0}^{k_1}...\sum_{k_{\delta-1}=0}^{k_{\delta-2}} \nonumber \\
&&\sum_{k_\delta=0}^{k_{\delta-1}}{b_{k_\delta}} {b_{{k_{\delta-1}}-{k_{\delta}}}} 
 {b_{{k_{\delta-2}}-{k_{\delta-1}}}}...{b_{{k_{2}}-{k_{3}}}}{b_{{k_{1}}-{k_{2}}}} {b_{n-k_1}} \Big) \Bigg].
\end{eqnarray}
The power series solution of Eq.(\ref{frac order ODE}) can be written in the form:
\begin{eqnarray}
  f(\xi)&=& b_0+b_1 \xi+ \sum_{n=0}^{\infty}b_{n+2} \xi^{n+2} \nonumber \\ 
 &=& b_0+b_1 \xi  \nonumber\\
&& +\sum_{n=0}^{\infty}\frac{1}{(n+1)(n+2)}\Bigg[\sum_{n=0}^{\infty}\frac{\Gamma(2-\frac{\alpha(1+2\delta)}{2\delta}+\frac{n\alpha}{2})}{\Gamma(2-\frac{\alpha(1+2\delta)}{2\delta}+\frac{n\alpha}{2}+\alpha)}b_n \nonumber \\
&& +\beta \Big(\sum_{k_1=0}^{n}\sum_{k_2=0}^{k_1}...\sum_{k_{\delta-1}=0}^{k_{\delta-2}}\sum_{k_\delta=0}^{k_{\delta-1}}{b_{k_\delta}} {b_{{k_{\delta-1}}-{k_{\delta}}}} {b_{{k_{\delta-2}}-{k_{\delta-1}}}}...{b_{{k_{2}}-{k_{3}}}}\nonumber\\
&& {b_{{k_{1}}-{k_{2}}}} (n-k_1+1){b_{(n-k_1+1)}} \Big) -\gamma b_n +\gamma \Big(\sum_{k_1=0}^{n}\sum_{k_2=0}^{k_1}...\sum_{k_{\delta-1}=0}^{k_{\delta-2}}\nonumber \\
&&\sum_{k_\delta=0}^{k_{\delta-1}}{b_{k_\delta}} {b_{{k_{\delta-1}}-{k_{\delta}}}} 
{b_{{k_{\delta-2}}-{k_{\delta-1}}}}...{b_{{k_{2}}-{k_{3}}}}{b_{{k_{1}}-{k_{2}}}} {b_{n-k_1}} \Big) \Bigg] \xi^{n+2}.
\end{eqnarray}
 Consequently, we acquire exact power series solution of Eq.(\ref{GBF}) as follows:
  \begin{eqnarray} \label{explicit power solution}
  u(x,t)&=&b_0 t^{-\frac{\alpha}{2}}+b_1 x t^{-\alpha} \nonumber\\
  &&+\sum_{n=0}^{\infty}\frac{1}{(n+1)(n+2)}\Bigg[\sum_{n=0}^{\infty}\frac{\Gamma(2-\frac{\alpha(1+2\delta)}{2\delta}+\frac{n\alpha}{2})}{\Gamma(2-\frac{\alpha(1+2\delta)}{2\delta}+\frac{n\alpha}{2}+\alpha)}b_n \nonumber \\
  && +\beta \Big(\sum_{k_1=0}^{n}\sum_{k_2=0}^{k_1}...\sum_{k_{\delta-1}=0}^{k_{\delta-2}}\sum_{k_\delta=0}^{k_{\delta-1}}{b_{k_\delta}} {b_{{k_{\delta-1}}-{k_{\delta}}}} {b_{{k_{\delta-2}}-{k_{\delta-1}}}}...{b_{{k_{2}}-{k_{3}}}}{b_{{k_{1}}-{k_{2}}}} 
  \nonumber \\
&& (n-k_1+1){b_{(n-k_1+1)}} \Big) -\gamma b_n +\gamma \Big(\sum_{k_1=0}^{n}\sum_{k_2=0}^{k_1}...\sum_{k_{\delta-1}=0}^{k_{\delta-2}}\sum_{k_\delta=0}^{k_{\delta-1}}{b_{k_\delta}} 
\nonumber \\
&&{b_{{k_{\delta-1}}-{k_{\delta}}}} {b_{{k_{\delta-2}}-{k_{\delta-1}}}}...{b_{{k_{2}}-{k_{3}}}}{b_{{k_{1}}-{k_{2}}}} {b_{n-k_1}} \Big) \Bigg]x^{n+2}  t^{-\frac{\alpha(n+3)}{2}}.
 \end{eqnarray}

\section{Convergence analysis}
In this section, convergence of the power series  solution of  Eq.(\ref{explicit power solution}) will be presented. Consider Eq.(\ref{bnplustwo}), we can write
\begin{eqnarray} \label{convergebn}
    |{b_{n+2}}| &\leq&  \Bigg \{ \frac{|\Gamma(2-\frac{\alpha(1+2\delta)}{2\delta}+\frac{n\alpha}{2})|}{|\Gamma(2-\frac{\alpha(1+2\delta)}{2\delta}+\frac{n\alpha}{2}+\alpha)|}|b_n| +|\beta| \Big(\sum_{k_1=0}^{n}\sum_{k_2=0}^{k_1}...\sum_{k_{\delta-1}=0}^{k_{\delta-2}}\sum_{k_\delta=0}^{k_{\delta-1}}{|b_{k_\delta}|} \nonumber \\
&&{|b_{{k_{\delta-1}}-{k_{\delta}}}|} {|b_{{k_{\delta-2}}-{k_{\delta-1}}}|}...{|b_{{k_{2}}-{k_{3}}}|}{|b_{{k_{1}}-{k_{2}}}|} {|b_{(n-k_1+1)}|} \Big) -|\gamma| |b_n| \nonumber \\
&&+|\gamma| \Big(\sum_{k_1=0}^{n}\sum_{k_2=0}^{k_1}...\sum_{k_{\delta-1}=0}^{k_{\delta-2}}\sum_{k_\delta=0}^{k_{\delta-1}}{|b_{k_\delta}|} {|b_{{k_{\delta-1}}-{k_{\delta}}}|} 
\nonumber \\
&&{|b_{{k_{\delta-2}}-{k_{\delta-1}}}|}...{|b_{{k_{2}}-{k_{3}}}}{b_{{k_{1}}-{k_{2}}}|} {|b_{n-k_1}|} \Big) \Bigg\}.
\end{eqnarray}
 It is well known from  the properties of $\Gamma$, it is easily found that \\
$\frac{|\Gamma(2-\frac{\alpha(1+2\delta)}{2\delta}+\frac{n\alpha}{2})|}{|\Gamma(2-\frac{\alpha(1+2\delta)}{2\delta}+\frac{n\alpha}{2}+\alpha)|} < 1$ for arbitrary $n$.\\
Thus Eq.(\ref{convergebn}) can be written as 
\begin{eqnarray}
    {|b_{n+2}|} &\leq& M \Bigg \{ |b_n| + \Big(\sum_{k_1=0}^{n}\sum_{k_2=0}^{k_1}...\sum_{k_{\delta-1}=0}^{k_{\delta-2}}\sum_{k_\delta=0}^{k_{\delta-1}}{|b_{k_\delta}|} {|b_{{k_{\delta-1}}-{k_{\delta}}}|} \nonumber \\
&&{|b_{{k_{\delta-2}}-{k_{\delta-1}}}|}...{|b_{{k_{2}}-{k_{3}}}|} {|b_{{k_{1}}-{k_{2}}}|} {|b_{(n-k_1+1)}|} \Big) \nonumber\\
&&+ \Big(\sum_{k_1=0}^{n}\sum_{k_2=0}^{k_1}...\sum_{k_{\delta-1}=0}^{k_{\delta-2}}\sum_{k_\delta=0}^{k_{\delta-1}}{|b_{k_\delta}|} {|b_{{k_{\delta-1}}-{k_{\delta}}}|} 
\nonumber \\
&&{|b_{{k_{\delta-2}}-{k_{\delta-1}}}|}...{|b_{{k_{2}}-{k_{3}}}}{b_{{k_{1}}-{k_{2}}}|} {|b_{n-k_1}|} \Big) \Bigg\}
\end{eqnarray}
where $ M= max\{|(1-|\gamma|)|, |\gamma|,|\beta|\}$.\\
Consider another power series of the form 
\begin{equation}
    G(\xi)=\sum_{n=0}^{\infty}q_n \xi^n.
\end{equation}
Let $q_i=|b_i|, i=0,1,2.$ Then we have
\begin{eqnarray}
   q_{n+2} &\leq& M \Bigg \{ q_n + \Big(\sum_{k_1=0}^{n}\sum_{k_2=0}^{k_1}...\sum_{k_{\delta-1}=0}^{k_{\delta-2}}\sum_{k_\delta=0}^{k_{\delta-1}}{q_{k_\delta}} {q_{{k_{\delta-1}}-{k_{\delta}}}}  \nonumber \\
  &&{q_{{k_{\delta-2}}-{k_{\delta-1}}}}...{q_{{k_{2}}-{k_{3}}}}{q_{{k_{1}}-{k_{2}}}} {q_{(n-k_1+1)}} \Big) \nonumber\\
&& + \Big(\sum_{k_1=0}^{n}\sum_{k_2=0}^{k_1}...\sum_{k_{\delta-1}=0}^{k_{\delta-2}}\sum_{k_\delta=0}^{k_{\delta-1}}{q_{k_\delta}} {q_{{k_{\delta-1}}-{k_{\delta}}}} 
\nonumber \\
&&{q_{{k_{\delta-2}}-{k_{\delta-1}}}}...{q_{{k_{2}}-{k_{3}}}}{q_{{k_{1}}-{k_{2}}}} {q_{n-k_1}} \Big) \Bigg\}.
\end{eqnarray}
Therefore it is easily seen that $|q_n|\leq b_n, n=0,1,2,...$\\
On the other hand, the series $ G(\xi)=\sum_{n=0}^{\infty}q_n \xi^n$ is majorant series of Eq.(\ref{convergebn}). We next show that the series  $ G(\xi)$ has positive radius of convergence. By simple calculation, we have that 
\begin{eqnarray}
    G(\xi)&=& q_0+q_1 \xi+M\Bigg \{ \sum_{n=0}^{\infty} q_n + \Big(\sum_{n=0}^{\infty}\sum_{k_1=0}^{n}\sum_{k_2=0}^{k_1}...\sum_{k_{\delta-1}=0}^{k_{\delta-2}}\sum_{k_\delta=0}^{k_{\delta-1}}{q_{k_\delta}} {q_{{k_{\delta-1}}-{k_{\delta}}}} \nonumber \\
    && {q_{{k_{\delta-2}}-{k_{\delta-1}}}}...{q_{{k_{2}}-{k_{3}}}}{q_{{k_{1}}-{k_{2}}}} {q_{(n-k_1+1)}} \Big) 
     + \Big(\sum_{n=0}^{\infty}\sum_{k_1=0}^{n}\sum_{k_2=0}^{k_1}...\sum_{k_{\delta-1}=0}^{k_{\delta-2}}  \nonumber\\
    &&\sum_{k_\delta=0}^{k_{\delta-1}}{q_{k_\delta}} {q_{{k_{\delta-1}}-{k_{\delta}}}}{q_{{k_{\delta-2}}-{k_{\delta-1}}}}...{q_{{k_{2}}-{k_{3}}}}{q_{{k_{1}}-{k_{2}}}} {q_{n-k_1}} \Big) \Bigg\} \xi^{n+2}.
\end{eqnarray}

Consider an implicit functional system with respect to the independent variable $\xi$ as follows:
\begin{equation}
    \textbf{G}(\xi,G)=G-q_0-q_1\xi-M \{\xi^2 G+\xi G^\delta (G-q_0)+\xi^2 G^{\delta+1}\}
\end{equation}
Since $\textbf{G}$ is analytic in a neighbourhood of $(0,q_0)$ where $\textbf{G}(0,q_0)=0$ and $\frac{\partial}{\partial G}\textbf{G}(0,q_0)\neq 0.$ Then by implicit function theorem \cite{rudin}, one can see that the series  $ G(\xi)=\sum_{n=0}^{\infty}q_n \xi^n$ is analytic in neighbourhood of the point $(0,q_0)=0$ and with a positive radius. This implies that Eq.(\ref{explicit power solution}) converges in a neighbourhood of the point $(0,q_0)=0$.
\section{Conclusion}
In this paper, to generate all symmetries, we applied Lie group method for the time fractional generalized Burgers-Fisher equation using RL fractional derivative.  Using symmetries, we reduced the time fractional generalized Burgers-Fisher equation to the nonlinear fractional ODE. Then the power series has been obtained to get an explicit solution for the nonlinear fractional ODE of the generalized Burgers-Fisher equation with arbitrary nonlinearity and analyzed the convergence also. To illustrate the physical meaning of the exact solution, some plots are given with suitable parameters values. 

\begin{acknowledgements}
The authors thank the Department of Science and Technology-Fund Improvement of S\&T Infrastructure in Universities and Higher Educational Institutions Government of India (SR/FST/MSI-107/2015) for carrying out this research work.
\end{acknowledgements}



\begin{thebibliography}{}
\bibitem{podlubny}
 I. Podlubny,
 Fractional Differential Equations,
 Academic Press,
 NY,
 1999.

\bibitem{rudin}
 Walter Rudin,
 Principles of Mathematical Analysis,
 China Machine Press,
 Beijing,
 2004.












\bibitem{oldham}
 K. B. Oldham and J. Spanier,
 The fractional calculus,
 Academic Press,
 NY,
 1974.


\bibitem{samko}
 S. Samko and A. A. Kilbas and O. Marichev ,
 Fractional integrals and derivatives: Theory and applications,
 Gordon and Breach Science,
 Switzerland,
 1993.


 \bibitem{kilbas}
 A.Kilbas and H.Srivatsava and J.Trujillo,
 Theory and Applications of Fractional Differential Equations,
 North Holland
 NY,
 2006.
\bibitem{gazizov}
 R. K. Gazizov and A. A. Kasatkin and S. Y. Lukashcuk,
 Symmetry properties of fractional diffusion equations,
 Phys.Scr,
 136,
 2009,
 14-16.
 
 
 \bibitem{huang}
 Q. Huang and R. Zhdanov,
 Symmetries and exact solutions of the time fractional Harry-Dym equation with Riemann-Liouville derivative,
 Physica A,
 409,
 2014,
 110-118.
 
 
 
 \bibitem{Ilhan}
 O. A. Ilhan and A. Esen and H. Bulut  et al,
 Singular solitons in the pseudo-parabolic model arising in nonlinear surface waves,
  Results Phys.,
 12,
 2019,
 1712–1715.
 
 
 
 
  \bibitem{Bekir}
 A. Bekir,
 New exact traveling wave solutions of some complex nonlinear equations,
 Commun. Nonlinear Sci.,
 14,
 2019,
 1069–1077.
\bibitem{Mirzazadeh}
 M. Mirzazadeh and  M. E. Slami and E. Zerrad et al.,
 Optical solitons in nonlinear directional couplers by sine-cosine function method and Bernoulli’s equation approach,
 Nonlinear Dynam.,
 81,
 2015,
 1933–1949.

 

  \bibitem{He}
 J. H. He and  M. A. Abdou,
 New periodic solutions for nonlinear evolution equations using Exp-function method,
 Chaos Soliton. Fract.,
 34,
 2007,
 1421–1429.
 
 
 
 \bibitem{Naher}
 H. Naher and F. A. Abdullah,
 New traveling wave solutions by the extended generalized Riccati equation mapping method of the (2+1)-dimensional evolution equation,
 J. Appl. Math.,
 2012,
 2012,
 486458.
 

 \bibitem{AkbarNaher}
 H. Naher and F. A. Abdullah and M. A. Akbar,
 Generalized and improved $(G'/G)$-expansion method for (3+1) dimensional modified KdV-Zakharov-Kuznetsev equation,
 PLOS One,
 8,
 2013 ,
 64618.
 
 

\bibitem{mayelihosseini}
 K. Hosseini and P. Mayeli and D. Kumar,
 New exact solutions of the coupled Sine-Gordon equations in nonlinear optics using the modified Kudryashov method,
 J. Mod. Optic.,
 65,
 2018 ,
 361–364.
 

\bibitem{mayelihosseiniansari}
 K. Hosseini and P. Mayeli and R. Ansari,
 Modified Kudryashov method for solving the conformable Time-Fractional Klein-Gordon equations with Quadratic and cubic nonlinearities,
 Optik,
 130,
 2017,
 737–742.
 
\bibitem{korkmazhosseini}
 A. Korkmaz and  K. Hosseini,
 Exact solutions of a nonlinear conformable Time-Fractional parabolic equation with Exponential nonlinearity using reliable methods,
 Opt. Quant. Electron.,
 49,
 2017,
 278.

 
 
\bibitem{bejarbanehhosseini}
 K. Hosseini and E. Y. Bejarbaneh and A. Bekir et al
 New exact solutions of some nonlinear evolution equations of pseudo parabolic type
 Opt. Quant. Electron.
 49
 2017
 241.

\bibitem{Mahmud}
 F. Mahmud and M. Samsuzzoha and M. A. Akbar,
 The generalized kudryashov method to obtain exact traveling wave solutions of the PHI-four equation and the Fisher equation,
 Results Phys.,
 7,
 2017,
 4296–4302.

\bibitem{bibi}
  S. Bibi and N. Ahmed and U. Khan et al,
 Some new exact solitary wave solutions of the van der Walls model arising in nature,
 Results Phys.,
 9,
 2018,
  648–655,
 



\bibitem{kaplan}
  M. Kaplan and A. Bekir and A. Akbulut,
  A generalized kudryashov method to some nonlinear evolution equations in mathematical physics,
 Nonlinear Dynam.,
 85,
 2016,
 2843–2850,
 \bibitem{Demiray}
  S. T. Demiray and Y. Pandir and H. Bulut,
 Generalized Kudryashov method for Time-fractional differential equations,
 Abstr. Appl. Anal.,
 2014,
 2014,
 901540
 


\bibitem{Gepreel}
   K. A. Gepreel and T. A. Nofal and A. A. Alasmari,
 Exact solutions for nonlinear integro-partial differential equations using the generalized Kudryashov method,
 Journal of the Egyptian Mathematical Society,
 25,
 2017,
 438–444.


\bibitem{koparankaplan}
  M. Koparan and M. Kaplan and A. Bekir et al.,
  A novel generalized Kudryashov method for exact solutions of nonlinear evolution equations,
 AIP Conference Proceedings,
  1798,
 2017,
 020082.

\bibitem{bulutpandir}
 H. Bulut and   Y.Pandir and S. T.Demiray.,
  Exact Solutions of Time-Fractional KdV Equations by Using Generalized Kudryashov Method,
 Int. J. Model. Opt.,
 4,
 2014,
 315-320.

\bibitem{bulutbaskonus}
 H. Bulut and  H. M.Baskonus,  and Y.Pandir,
  The modified trial equation method for fractional wave equation and time fractional generalized Burgers equation,
 Abst. Applied Analy.,
 2013,
 1-8.

\bibitem{bulutbernoulli}
 H. Bulut and G. Yel and H. M.Baskonus,
  An application of Improved Bernoulli Sub-Equation method to the nonlinear time fractional generalized Burgers equation,
 Turk.J.Math.Comput.Sci,
  5,
 2015,
 1-7.

\bibitem{tarasov}
 Vasily E. Tarasov,
  On chain rule for fractional derivatives,
 Commun Nonlinear Sci Number Simulat,
  30,
 2016,
 1-4.

\bibitem{wu}
 C.C.Wu,
 A fractional Variational Iteration method for solving fractional nonlinear differential equations,
 Comput. Math. Appl.,
  61(8),
 2011,
 2186-2190.
\bibitem{elsayed}
 A. M. A. El-Sayed and M. Gaber,
 The Adomian decomposition method for solving partial differential equations of fractal order in finite domains,
 Phys. Lett. A,
  359(3),
 2006,
 175-182.

\bibitem{Gepreelhomotopy}
  K. A. Gepreel,
 The homotopy perturbation method applied to the nonlinear Kolmogorov-Petrovskii-Piskunov equations,
 Appl.Math.Lett,
 24(8),
 2009,
 1428–1434.

\bibitem{Lu}
  B. Lu,
 The first integral method for some time fractional differential equations,
 J.Math.Anal.Appl.,
 395(2),
 2012,
 684-693.

\bibitem{bekirunsal}
  A. Bekir and O. Guner and O.Unsal,
 The first integral method for exact solutions of nonlinear fractional differential equations,
 J.Comp.Nonlinear Dyn,
 10(2),
 2014,
 5 pages.

\bibitem{tong}
 B. Tong and Y. He and L. Wei and X. Zhang,
  Generalized Fractional Sub-Equation Method for Fractional Differential Equations With Variable Coefficients,
 Phys. Lett. A,
 376(38–39),
 2012,
 2588–2590.



\bibitem{zhangzhang}
 S. Zhang and H. Q. Zhang,
 Fractional Sub-Equation Method and Its Applications to Nonlinear Fractional PDEs,
 Phys. Lett. A,
 375(7),
 2011,
 1069–1073.


\bibitem{veeresha}
 P. Veeresha and D. G. Prakasha and H. M. Baskonus,
 Novel simulations to the time fractional Fisher's equation,
 Mathematical Sciences,
 13,
 2019,
 33-42.



\bibitem{kkr1}
 K. M. Tamizhmani and K. Krishnakumar and P. G L. Leach,
 Algebraic resolution of equations of the Black–Scholes type with arbitrary time-dependent parameters,
 Applied Mathematics and Computation,
 247,
 2014,
 115-124.
\bibitem{kkr2}
 A. Paliathanasis and K. Krishnakumar and K. M.  Tamizhmani and  P. G. L.  Leach,
 Algebraic resolution of equations of the Black–Scholes type with arbitrary time-dependent parameters,
 Mathematics,
 4(28),
 2016,
 1-14.


\bibitem{wang2}
 Gangwei Wang and Tianzhou Xu,
 Invariant analysis and explicit solutions of the time fractional nonlinear perturbed Burgers equation,
 Nonlinear Analysis: Modelling and Control,
 20(4),
 2015,
 570-584.
 \bibitem{zhangyw}
 Y.W. Zhang,
 Lie symmetry analysis to generalized fifth-order time-fractional KdV equation,
 Nonlinear Studies,
 22(3),
 2015,
 473-484.

\bibitem{juma}
G.Jumarie,
Cauchy's integral formula via modified riemann-Liouville derivative for analytic functions of fractional order,
 Appl.Math.Letter,
 23,
2010,
 1444 - 1450.
 
\bibitem{eslami},
M.Eslami,
Exact traveling wave solutions to the fractional coupled nonlinear Schrodinger equations ,
Appl.Math.Comput.,
285,
 2016,
141 - 148.
 
\bibitem{baakkiyaraj},
]R.Sahadevan  and T.Bakkyaraj,
 Invariant analysis of time fractional generalized Burgers and Korteweg-de Vries equations,
 J.Math.Anal.Appl.,
393,
2012,
341 - 347.
\bibitem{huan},
Huang Q and Zhdanov R.,
 Symmetries and exact solutions of the time fractional Harry-Dym equation with Riemann-Liouville derivative,
Phys. A.,
 409,
2014,
110 - 118.
\bibitem{Guo},
S. Guo and Y. Mei and Y. Li and Y. Sun,
The improved fractional sub-equation method and its applications to the space-time fractional differential equations in fluid mechanics,
Phys. Letters A.,
 376,
 2012,
407-411.
\bibitem{kirya},
V.Kiryakova,
Generalized Fractional Calculus and Applications,
Pitman Research Notes in Mathematics,
 301,
 1994.
 
 \bibitem {Alzaidy}
 J. F. Alzaidy,
Fractional Sub-Equation Method and its Applications to the Space-Time Fractional Differential Equations in Mathematical Physics,
Br. J. of Maths. Comp. Sci.,
 2,
 2013,
152-163.
\bibitem{lie}
]S.Lie,
 On integration of a class of linear partial differential equations by means of definite integrals,
 Arch.Math.Log,
6(3),
1881,
328 - 368.
 
\bibitem{ovsiannikov}
 L.V.Ovsiannikov ,
Group analysis of differential equations,
 Academic Press, New York,
1982.
\bibitem{bal17}
D. Baleanu and M.Inc and A. Yusuf and A. I. Aliyu,
Lie symmetry analysis, exact solutions and conservation laws for the time fractional modified Zakharov-Kuznetsov equation,
 Nonlinear Anal.Model.Control,
22(6),
2017,
 861 - 876.
 



\bibitem{bal18}
D. Baleanu and M. Inc and A. Yusuf and A. I. Aliyu,
Lie symmetry analysis, exact solutions and conservation laws for the time fractional Caudrey - Dodd - Gibbon - Sawada - Kotera equation,
Commun.Nonlinear Sci.Number.Simul.,
 59,
2018,
222 - 234.
 
 
 
 \bibitem{kkr}
 K.M. Tamizhmani, K. Krishnakumar and P.G.L.Leach,
 Symmetries and reductions of order for certain nonlinear third and second order differential equations with arbitrary nonlinearity,
 Journal of Mathematical Physics,
56,
2015,
113503-1 - 113503-11.
 
\bibitem{bluman}
G.W.Bluman and S.Kumei,
Symmetries and Differential Equations,
Springer Verlag,
NY,
1989.

\bibitem{hydon}
Peter E. Hydon,
Symmetry Methods for Differential Equations:A Beginner's Guide,
 Cambridge University Press,
 UK,
2000.
 
 
 
\bibitem{jeff},
G.F. Jefferson and Carminati.J,
FracSym: Automated symbolic computation of Lie symmetries of fractional differential equations,
Comput.Phys.Commun.,
185,
2014,
430 - 441. 
 
\bibitem{olver}
P.J.Olvar,
Applications of Lie Groups to Differential Equations,
Springer Verlag,
NY,
1993.

 
 
\bibitem{ibragi}
 N.H.Ibragimov,
Handbook of Lie Group Analysis of Differential Equations,
 CRC Press,
 1,
 Boca Raton,
1994.
 
 
\bibitem{cradd}
M.Craddock and K.Lennox,
Lie symmetry methods for multi-dimensional parabolic PDEs and diffusions,
 J.Differ.Equ.,
 252,
 2012,
56 - 90.
\bibitem{mustafa}
 M.Inc and A.Yusuf and A.I.Aliyu and D. Baleanu,
Lie symmetry analysis and explicit solutions for the time fractional generalized Burgers-Huxley equation,
 Opt.Quant.Electron,
50(94),
2018,
1 - 16.

\bibitem{wang}
Wang GW and Xu TZ,
Symmetry properties and explicit solutions of the nonlinear time fractional KdV equation,
 Bound. Value Probl,
 232,
2013,
232 - 245.
 
\bibitem{kumar}
Vikas Kumar and Lakhveer Kaur and Ajay Kumar and Mehmet Emir Koksal,
Lie symmetry based analytical and numerical approach for modified Burgers-KdV equation,
Results in Physics,
8,
2018,
1136 - 1142.

\bibitem{baksaha}
 T.Bakkyaraj and R.Sahadevan,
Group formalism of Lie transformations to time-fractional partial differentiall equations,
Pramana- Journal of Physics,
 85,
 5,
Nov 2015,
849 - 860.
  
\bibitem{saha}
T.Bakkyaraj and R.Sahadevan,
Invariant analysis of nonlinear fractional ordinary differential equations with Riemann-Liouville fractional derivative,
Nonlinear Dynam.,
 80(1-2),
2015,
 447 - 455.


 
 
 
 
 
 
 
 
 
 
 
 
 
 
 
 
 
 
 
 
 
 
%

\end{thebibliography}


\end{document}